\def\rank{\mathop{\rm rank}}
\def\wgt{\mathop{\rm wgt}}

\def\ker{\mathop{\rm Ker}}
\def\im{\mathop{\rm Im}}
\documentclass[aps,%prb,12pt,tightenlines,%
prl,twocolumn,%superscriptaddress,%
notitlepage,longbibliography]{revtex4-1}
\usepackage{hyperref}
\usepackage{amsthm}
\usepackage{amsfonts}
\newtheorem{theorem}{Theorem}

\begin{document}
\title{Higher-dimensional quantum hypergraph-product codes}
\date\today
%\date{\today\ \bf \jobname} 
%\date{\today\ \currenttime\ \bf \jobname} 
\author{Weilei Zeng}
\address{Department of Physics \& Astronomy, University of California,
  Riverside, California 92521, USA}
\author{Leonid P.\ Pryadko}
\address{Department of Physics \& Astronomy, University of California,
  Riverside, California 92521, USA}
\email{leonid.pryadko@ucr.edu}

\begin{abstract}
  We describe a family of quantum error-correcting codes which
  generalize both the quantum hypergraph-product (QHP) codes by
  Tillich and Z\'emor, and all families of toric codes on
  $m$-dimensional hypercubic lattices.  Similar to the latter, our
  codes form $m$-complexes ${\cal K}_m$, with $m\ge2$.  These are
  defined recursively, with ${\cal K}_m$ obtained as a tensor product
  of a complex ${\cal K}_{m-1}$ with a $1$-complex parameterized by a
  binary matrix.  Parameters of the constructed codes are given
  explicitly in terms of those of binary codes associated with the
  matrices used in the construction.
\end{abstract}
\maketitle

Quantum low-density parity-check (q-LDPC) codes is the only class of
codes known to combine finite rates with non-zero fault-tolerant (FT)
thresholds\cite{Kovalev-Pryadko-FT-2013,Dumer-Kovalev-Pryadko-bnd-2015},
to allow scalable quantum computation with a finite
overhead\cite{Gottesman-overhead-2014}.  However, unlike in the
classical case where capacity-approaching codes can be constructed
from random sparse matrices\cite{Gallager-1962,Gallager-book-1963,%
  Litsyn-Shevelev-2002,%
  Richardson-Shokrollahi-Amin-Urbanke-2001}, matrices suitable for
constructing quantum LDPC codes are highly atypical in the
corresponding ensembles.  Thus, an algebraic ansatz is required to
construct large-distance q-LDPC codes.  Preciously few examples of
such algebraic constructions are known that give finite rate codes and
also satisfy conditions\cite{Dumer-Kovalev-Pryadko-bnd-2015} for
fault-tolerance: bounded weight of stabilizer generators and minimum
distance that scales logarithmically or faster with the block length
$n$.  Such constructions include hyperbolic codes in
two\cite{Zemor-2009,Delfosse-Zemor-2010,Breuckmann-Terhal-2015,%
  Breuckmann-Vuillot-Campbell-Krishna-Terhal-2017} and higher
dimensions\cite{Guth-Lubotzky-2014}, and quantum hypergraph-product
(QHP) and related codes\cite{Tillich-Zemor-2009,Tillich-Zemor-2014,%
  Kovalev-Pryadko-Hyperbicycle-2013}.  In addition, some
constructions, e.g., in Refs.\
\onlinecite{Couvreur-Delfosse-Zemor-2012,Bravyi-Hastings-2013,Audoux-2014,%
  Audoux-Couvreur-2017,Hastings-codes-2016}, have finite rates and
relatively high distances, with the stabilizer generator weights that
grow with $n$ logarithmically.  It is not known whether these codes
have non-zero FT thresholds.  However, such codes can be modified into those
with provable FT thresholds with the help of weight
reduction\cite{Hastings-weight-2016}.

There is more variety for topological codes, which can be viewed as
generalized toric codes\cite{Bravyi-Kitaev-1998,Freedman-Meyer-1998,%
  Dennis-Kitaev-Landahl-Preskill-2002,Bombin-MartinDelgado-2007,%
  Castelnovo-Chamon-2008,Mazac-Hamma-2012,%
  Bombin-Chhajlany-Horodecki-MartinDelgado-2013} invented by
Kitaev\cite{kitaev-anyons}.  Such a code can be constructed from any
tessellation of an arbitrary surface or a higher-dimensional manifold.
Essential advantage of topological codes is locality: each stabilizer
generator involves only the qubits in the immediate vicinity of each
other; it is this feature that makes planar surface codes so
practically attractive.  However, locality also limits the parameters
of topological
codes\cite{Bravyi-Terhal-2009,Bravyi-Poulin-Terhal-2010,%
  Delfosse-2013,Flammia-Haah-Kastoryano-Kim-2017}.  In particular, for
a code of length $n$ with stabilizer generators local in two
dimensions, the number of encoded qubits $k$ and the minimal distance
$d$ satisfy the inequality\cite{Bravyi-Terhal-2009}
$kd^2\le {\cal O}(n)$.  This implies asymptotically zero rate whenever
$d$ diverges with $n$.

In this work we construct a family of q-LDPC codes that generalize the
QHP codes\cite{Tillich-Zemor-2009,Tillich-Zemor-2014} to higher
dimensions, and explicitly calculate their parameters, including the
minimum distances.  Our codes relate to toric codes on hypercubic
lattices\cite{Dennis-Kitaev-Landahl-Preskill-2002,Bombin-MartinDelgado-2007,%
  Castelnovo-Chamon-2008,Mazac-Hamma-2012,%
  Bombin-Chhajlany-Horodecki-MartinDelgado-2013} in exactly the same
fashion as the QHP codes relate to the square-lattice toric code.
Just as different $m$-dimensional toric codes on a hypercubic lattice
are parts of an $m$-complex\cite{Bombin-MartinDelgado-2007}, here we
also construct $m$-complexes, chain complexes with $m$ non-trivial
boundary operators.  Our construction is recursive: it defines an
$m$-complex ${\cal K}_{m}$ as a tensor product of a shorter chain
complex ${\cal K}_{m-1}$ and a $1$-complex ${\cal K}_1$, a linear map
between two binary vector spaces.  In particular, the construction of
the $2$-complex ${\cal K}_2$ in terms of two binary matrices is
identical to QHP
codes\cite{Tillich-Zemor-2009,Tillich-Zemor-2014}.

Previously, related constructions have been considered in
Refs.~\onlinecite{Hastings-weight-2016,Audoux-Couvreur-2017,Campbell-2018}.
Hastings\cite{Hastings-weight-2016} only considered products with
1-complexes which correspond to classical repetition codes, in
essence, the same construction that appears in ``space-time'' codes
used in the analysis of repeated syndrome
measurement\cite{Wang-Harrington-Preskill-2003,%
  Kovalev-Pryadko-FT-2013,Dumer-Kovalev-Pryadko-bnd-2015}.  On the
other hand, Audoux and Couvreur\cite{Audoux-Couvreur-2017} and
Campbell\cite{Campbell-2018} only considered products of
$2$-complexes.  Their lower bounds on code distances are not generally
as strong as ours.

In addition to defining new classes of quantum LDPC codes with
parameters known explicitly, our construction may be useful for
optimizing repeated measurements in the problem of fault-tolerant (FT)
quantum error correction, related problem of single-shot error
correction\cite{Fujiwara-2014,Ashikhmin-Lai-Brun-2014,%
  Ashikhmin-Lai-Brun-2016,Campbell-2018},
analysis of transformations between different QECCs, like the
distance-balancing trick by Hastings\cite{Hastings-weight-2016}, and
construction of asymmetric quantum CSS codes optimized for operation
where error rates for $X$ and $Z$ channels may differ
strongly\cite{Ioffe-Mezard-2007,Evans-2007,Stephens-2008,%
  Aliferis-Preskill-2008,sarvepalli-2009,Tuckett-Bartlett-Flammia-2018}.

We start with a brief overview of error correcting codes and chain
complexes, see, e.g., Refs.\
\onlinecite{Weibel-book-1994,Bombin-MartinDelgado-2007,%
  Audoux-Couvreur-2017,MS-book,gottesman-thesis,%
  Nielsen-book,Calderbank-1997} for much more information.  A
classical binary linear code ${\cal C}$ with parameters $[n,k,d]$ is a
$k$-dimensional subspace of the vector space $\mathbb{F}_2^n$ of all
binary strings of length $n$.  Code distance $d$ is the minimal
Hamming weight of a nonzero string in the code.  A code
${\cal C} \equiv {\cal C}_G$ can be specified in terms of the
generator matrix $G$ whose rows are the basis vectors of the code.
All vectors orthogonal to the rows of $G$ form the dual code
${\cal C}^\perp_G = \{c\in \mathbb{F}_n^2\vert Gc^T = 0\}$.  Matrix
$G$ is called the parity check matrix of the code ${\cal C}^\perp_G$.

Given an \emph{index set} $I\subseteq \{1,2,\ldots,n\}$ of length
$|I|=r$, and a string $c\in\mathbb{F}_2^n$, let
$c[I]\in\mathbb{F}_2^r$ be a substring of $c$ with the bits at all
positions $i\not\in I$ dropped.  Similarly, for an $n$-column matrix
$G$ with rows $g_j$, $G[I]$ is formed by the rows $g_j[I]$.  If
${\cal C}={\cal C}_G$ is a linear code with the generating matrix $G$,
the \emph{punctured} code ${\cal C}_p[I]\equiv\{c[I]: c\in{\cal C}\}$
is a linear code of length $|I|$ with the generating matrix $G[I]$.
The \emph{shortened} code ${\cal C}_s[I]$ is formed similarly, except
only from the codewords which have all zero bits outside $I$,
${\cal C}_s[I]=\{c[I]:c=(c_1,c_2,\ldots,c_n)\in {\cal C}$ and $c_i=0$
for each $i\not\in I\}$.  If ${\cal C}={\cal C}_P^\perp$ has the
parity check matrix $P$, $P[I]$ is the parity check matrix of the
shortened code ${\cal C}_s[I]$.

A \emph{chain complex} is a sequence of finite-dimensional vector
spaces $\ldots,{\cal A}_{j-1},{\cal A}_j,\ldots $ with \emph{boundary}
operators $\partial_j:{\cal A}_{j-1}\leftarrow {\cal A}_j $ that map
between each pair of neighboring spaces, with the requirement
$\partial_j\partial_{j+1}=0$, $j\in\mathbb{Z}$.  In this work we only
consider vector spaces ${\cal A}_j=\mathbb{F}_2^{n_j}$ formed by
binary vectors of length $n_j\ge0$, and define an $m$-complex
${\cal A}\equiv {\cal K}(A_1,\ldots,A_m)$, a length-$(m+1)$ chain
complex with a basis, in terms of $n_{j-1}\times n_j$ binary matrices
$A_j$ serving as the boundary operators,
\begin{equation}
  \label{eq:chain-complex}
{\cal A}:\;
  \{0\}\stackrel{\partial_0}\leftarrow {\cal A}_0\stackrel{A_1}\leftarrow
  {\cal A}_1\ldots \stackrel{A_{m}}\leftarrow
  {\cal A}_{m}\stackrel{\partial_{m+1}}\leftarrow \{0\},
\end{equation}
where the neighboring matrices must be mutually orthogonal,
$A_{j-1}A_{j}=0$, $j\in\{1,\ldots,m\}$.  In addition to boundary
operators given by the matrices $A_j$, implicit are the trivial
operators $\partial_0:\{0\}\leftarrow {\cal A}_0$ and
$\partial_{m+1}: {\cal A}_m\leftarrow \{0\}$ treated formally as zero
$0\times n_0$ and $n_m\times 0$ matrices.

Elements of the subspace $\im (\partial_{j+1})\subseteq {\cal A}_j$ are
called boundaries; in our case these are linear combinations of columns
of $A_{j+1}$ and, therefore, form a binary linear code
with the generator matrix $A_{j+1}^T$,
$\im (A_{j+1})=\mathcal{C}_{A_{j+1}^T}$.  In the singular case $j=m$,
$\im(\partial_{m+1})=\{0\}$, a trivial vector space.  Elements of
$\ker(\partial_j)\subset {\cal A}_j$ are called cycles; in our case these are vectors
$x$ in ${\cal A}_j$ orthogonal to the rows of $A_{j}$, $A_jx^T=0$.  This
defines a binary linear code with the parity check matrix $A_j$,
$\ker (A_j)=\mathcal{C}_{A_j}^\perp$.  In the singular case $j=0$,
$\ker(\partial_0)={\cal A}_0$.

Because of the orthogonality $\partial_j\partial_{j+1}=0$, all
boundaries are necessarily cycles,
$\im(\partial_{j+1})\subseteq \ker(\partial_j) \subseteq {\cal A}_j$.
The structure of the cycles in ${\cal A}_j$ that are not boundaries is
described by the $j$\,th homology group,
\begin{equation}
  {H}_j({\cal A})\equiv H(A_j,A_{j+1})=
  \ker(A_{j})/\im(A_{j+1}).
\label{eq:homo-group}
\end{equation}
Group quotient here means that two cycles [elements of $\ker (A_j)$]
that differ by a boundary [element of $\im(A_{j+1}$] are considered
equivalent; non-zero elements of $\mathcal{H}_j(\mathcal{A})$ are
equivalence classes of homologically non-trivial cycles.  We denote
the equivalence as $x\stackrel{A_{j+1}}\simeq y\in {\cal A}_j$, or just
$x\simeq y$.  Explicitly,
this implies that for some $\alpha\in {\cal A}_{j+1}$, $y=x+A_{j+1}\alpha$.
The rank of $j$-th homology group is the dimension of the
corresponding vector space; one has
\begin{equation}
  \label{eq:homo-rank}
  k_j\equiv \rank H_j(\mathcal{A})=n_j-\rank A_j-\rank A_{j+1}  .
\end{equation}
The homological \emph{distance} $d_j$ is the minimum Hamming weight of
a non-trivial element (any representative) in the homology group
$H_j(\mathcal{A})\equiv H(A_j,A_{j+1})$,
\begin{equation}
  d_j=\min_{ 0\not\simeq x\in H_j(\mathcal{A})} \wgt x
  =\min_{x\in \ker({A}_j)\setminus \im(A_{j+1})}\wgt x.
  \label{eq:homo-distance}
\end{equation}
By this definition, $d_j\ge1$.  To address singular cases, throughout
this work we assume that the minimum of an empty set is an infinity;
$k_j=0$ always implies $d_j=\infty$.

For an alternative definition, the rightmost expression in
Eq.~(\ref{eq:homo-distance}) treats vector spaces as sets.  Thus, to
calculate the distance $d_0$ of the homology group $H_0(\mathcal{A})$,
we have to take the minimum weight of all vectors $x\in C_0$ except
those that can be obtained as linear combinations of columns of $A_1$
[these form a binary linear code\cite{MS-book} $\mathcal{C}_{A_1^T}$
with the generator matrix $A_1^T$].  The result is $d_0=1$, unless
$A_1$ has a full row rank, giving $k_0=0$, in which case our
convention gives $d_0=\infty$.  

Similarly, in the case of the homology group $H_{m}(\mathcal{A})$, the
distance $d_{m}$ is the minimum weight of a non-zero $x\in C_m$ such
that $A_{m}x^T=0$.  In this case $d_{m}$ is also the distance of a
binary classical code $\mathcal{C}^\perp_{A_m}$ with the parity check
matrix $A_m$.  Again, our convention gives $d_m=\infty$ if $k_m=0$,
which happens when $A_m$ has full column rank.

In addition to the homology group $H(A_j,A_{j+1})$, there is also a
generally distinct \emph{co-homology} group
$\tilde{H}_j(\tilde{\cal A})=H(A_{j+1}^T,A_j^T)$ of the same rank
(\ref{eq:homo-rank}); this is associated with the \emph{co-chain
  complex} $\tilde{A}$ formed from the transposed matrices $A_j^T$
taken in the opposite order.  A quantum Calderbank-Shor-Steane (CSS)
code\cite{Calderbank-Shor-1996,Steane-1996} with generator matrices
$G_X=A_j$ and $G_Z=A_{j+1}^T$ is isomorphic with the direct sum of the
groups $H_j$ and $\tilde{H}_j$,
\begin{equation}
  \label{eq:css-code}
  \mathcal{Q}(A_j,A_{j+1}^T)\cong H(A_j,A_{j+1})\oplus H(A_{j+1}^T,A_j^T).
\end{equation}
The two terms correspond to $Z$ and $X$ logical operators,
respectively.  The code distance can be expressed as a minimum over
the distances $d_j$ and $\tilde{d}_j$ of the two homology groups.
Parameters of such a code are written as
$[[n_j,k_j,\min(d_j,\tilde{d}_{j})]]$.

Tensor product $\mathcal{A}\times \mathcal{B}$ of two chain complexes
$\mathcal{A}$ and $\mathcal{B}$ is defined as the chain complex formed
by linear spaces decomposed as direct sums of Kronecker products,
\begin{equation}
(\mathcal{A}\times \mathcal{B})_l=\bigoplus\nolimits_{i+j=l}\mathcal{A}_i \otimes
\mathcal{B}_{j},\label{eq:tp-spaces}
\end{equation}
with the action of the boundary operators 
\begin{equation}
  \label{eq:tp-boundary}
  \partial_{i+j}(a\otimes b)\equiv\partial_i' a\otimes b+(-1)^i a\otimes
  \partial_j'' b, 
\end{equation}
where $a\in\mathcal{A}_i$, $b\in\mathcal{B}_j$, and the boundary
operators $\partial_i'$ and $\partial_j''$ belong to the chain
complexes $\mathcal{A}$ and $\mathcal{B}$, respectively.  When both
$\mathcal{A}$ and $\mathcal{B}$ are \emph{bounded}, that is, they
include only a finite number of non-trivial spaces, the dimension
$n_j(\mathcal{C})$ of a space $\mathcal{C}_j$ in the product
$\mathcal{C}=\mathcal{A}\times \mathcal{B}$ is
\begin{equation}
  \label{eq:tp-nj}
n_j (\mathcal{C})=\sum\nolimits_{i} n_i (\mathcal{A}) \,n_{j-i}
  (\mathcal{B}).
\end{equation}
The homology groups of the product $\mathcal{C}={\cal A}\times {\cal B}$ are
isomorphic to a simple expansion in terms of those of ${\cal A}$ and
${\cal B}$ which is given by the K\"unneth theorem,
\begin{equation}
  \label{eq:tp-Kunneth}
  H_j (\mathcal{C})\cong\bigoplus\nolimits_{i} H_i (\mathcal{A})\,
  \otimes\,H_{j-i} 
  (\mathcal{B}).
\end{equation}
One immediate consequence is that the rank $k_j(\mathcal{C})$ of the
$j$\,th homology group $H_j(\mathcal{C})$ is 
\begin{equation}
  \label{eq:tp-kj}
  k_j (\mathcal{C})=\sum\nolimits_{i} k_i (\mathcal{A}) \,k_{j-i}
  (\mathcal{B}).
\end{equation}
%Further, in the case of chain complexes with the isomorphism 

Our first result is an upper bound on the distances of the
homological groups in a chain complex ${\cal A}\times {\cal B}$, an
immediate extension of Cor.\ 2.14 from
Ref.~\onlinecite{Audoux-Couvreur-2017},
\begin{equation}
  d_j(\mathcal{C})\le \min_i d_i({\cal A})\, d_{j-i}({\cal B}).
  \label{eq:tp-dj-upper-bnd}  
\end{equation}
\begin{proof}[Proof of Eq.~(\ref{eq:tp-dj-upper-bnd}).]
  This is a consequence of a version of the K\"unneth theorem for a
  pair of chain complexes with chosen bases, see Proposition 1.13 in
  Ref.~\onlinecite{Audoux-Couvreur-2017}. Namely, if, for each
  $r\in\mathbb{Z}$, the sets $X_r\subset {\cal A}_r$ and
  $Y_r\subset {\cal B}_r$ induce bases for $H_r({\cal A})$ and
  $H_r({\cal B})$, respectively, then, for every $j\in\mathbb{Z}$, the
  vectors in the set
  \begin{equation}
    Z_j=\{x\otimes y|i\in\mathbb{Z},x\in X_i,y\in Y_{j-i}\}
    \label{eq:tp-basis}
  \end{equation}
  induce a basis for $H_j({\cal A}\otimes {\cal B} )$.  Now, if we
  choose each of the sets $X_r$ and $Y_r$ to contain the corresponding
  minimum-weight vectors, minimum weight of the elements of the set
  (\ref{eq:tp-basis}) equals to the r.h.s.\ in
  Eq.~(\ref{eq:tp-dj-upper-bnd}).  The homology group is trivial,
  $k_j({\cal A}\otimes {\cal B})=0$ and $Z_j=\emptyset$, only if at
  least one of the sets in each pair $\{a_i,b_{j-i}\}$,
  $i\in\mathbb{Z}$ is empty, which implies that the corresponding
  product $d_i({\cal A})d_{j-i}({\cal A})$ be infinite, consistent
  with the result given by our convention, $d_j({\cal C})=\infty$
  whenever $k_j({\cal C})=0$.
\end{proof}
Our second result is a lower bound on the distance for the special
case where ${\cal B}={\cal K}(P)$ is a $1$-complex induced by an
$r\times c$ binary matrix $P$.  This bound matches the upper bound in
Eq.~(\ref{eq:tp-dj-upper-bnd}), and thus ensures the equality for the
case where ${\cal B}$ is a $1$-complex.  This expression,
\begin{equation}
  \label{eq:tp-dj-equality}
  d_j(\mathcal{A}\times {\cal B})=\min \biglb( d_{j-1}({\cal A})\,d_1({\cal
    B}), d_{j}({\cal  A})\,d_0({\cal B}) \bigrb),
\end{equation}
where $ {\cal B}={\cal K}(P)$ is a $1$-complex, is our main result.

With $\mathcal{A}$ the $m$-complex in
Eq.~(\ref{eq:chain-complex}), the tensor product
$ \mathcal{C}\equiv \mathcal{A}\times \mathcal{B}$ can be written as
an $(m+1)$-complex, $\mathcal{C}=\mathcal{K}(C_1,\ldots,C_{m+1})$,
with the block matrices
\begin{equation}
  \label{eq:tp-matrices}
      C_{j+1}=\left(
           \begin{array}[c]{c|c}
             A_{j+1}\otimes E_r&(-1)^{j}E_{n_{j}}\otimes P\\ \hline
                           & A_{j}\otimes E_c
           \end{array}\right), 
\end{equation}
where $E_r$ denotes the $r\times r$ identity matrix.  The sign in the
top-right corner ensures orthogonality $C_j C_{j+1}=0$; in our case
signs have no effect since we are only considering binary spaces.  We
also notice that since $\partial_0$ and $\partial_{m+1}$ in
$\mathcal{A}$ are both trivial, matrices $C_1$ and $C_{m+1}$,
respectively, will be missing the lower and the left block pairs.  If
we denote $u\equiv \rank P$, the two  homology groups
associated with $\mathcal{B}$ have ranks
${\kappa}_0\equiv k_0(\mathcal{B})=r-u$ and
${\kappa}_1\equiv k_1(\mathcal{B})=c-u$, respectively.  Equations
(\ref{eq:tp-nj}) and (\ref{eq:tp-kj}) give in this case,
\begin{equation}
  \label{eq:stp-n-and-k}
  n_j'=n_{j-1} c+n_{j}r\ \text{\ and\ }\ k_j'= k_{j-1}{\kappa}_1+k_{j}\kappa_0,
\end{equation}
where we use the primes to denote the parameters of ${\cal C}$,
$n_j'\equiv n_j({\cal C})$ and $k_j'\equiv k_j({\cal C})$.  %%% The first
%%% expression is self-evident.
% A self-contained derivation of the
% expression for $k_j'$ is given in the online supplement.
We now prove the claimed lower bound for the distance:
\begin{theorem}
  \label{th:lower-distance-bnd}
  Consider $m$-complex ${\cal A}$ in Eq.~(\ref{eq:chain-complex}), and
  assume that homological groups $H_j({\cal A})$ have distances $d_j$,
  $0\le j \le m$.  Given an $r\times c$ binary matrix $P$ of rank $u$,
  construct matrices $C_j$ in Eq.~(\ref{eq:tp-matrices}).  Denote
  $\delta$ the minimum distance of a binary code with the parity check
  matrix $P$; by our convention, $\delta=\infty$ if $u=c$.  The
  minimum distance $d_j'\equiv d_j({\cal C})$ of the homology group
  $H(C_j,C_{j+1})$, $0\le j\le m+1$,
  satisfies the following lower bounds:\\ \emph{(i)} if $r>u$, 
  $d_j'\ge \min(d_j,d_{j-1}\delta)$, otherwise,\\
  \emph{(ii)}  if $r=u$,
  $d_j'\ge d_{j-1}\delta$.
\end{theorem}
%%% For the reference, we are looking at  
%%% \begin{eqnarray*}
%%%   C_j&=&\left(
%%%     \begin{array}[c]{c|c}
%%%       A_j\otimes E_r&E_{n_{j-1}}\otimes P\\ \hline 
%%%                     & A_{j-1}\otimes E_c
%%%     \end{array}\right),\\
%%%   C_{j+1}&=&\left(
%%%     \begin{array}[c]{c|c}
%%%       A_{j+1}\otimes E_r&       E_{n_{j}}\otimes P\\ 
%%%       \hline 
%%%                         & A_{j}\otimes E_c
%%%     \end{array}\right).
%%% \end{eqnarray*}
\begin{proof}%[Proof of Theorem \ref{th:lower-distance-bnd}]
  Start with  (i). Take a block vector ${e}=({e}_1|{e}_2)$, with
  ${e}_1\in\mathbb{F}_2^{n_j r}$, ${e}_2\in\mathbb{F}_2^{n_{j-1} c}$,
  with component weights $w_1\equiv \wgt({e}_1)<d_j$, and
  $w_2\equiv \wgt({e}_2)<d_{j-1}\delta$, and assume $C_j {e}^T=0$.  We
  are going to show that ${e}$ is a linear combination of columns of
  $C_{j+1}$. 

  Step 1: This step is needed if $d_j$ is finite; otherwise let
  $C_j'=C_j$, $C_{j+1}'=C_{j+1}$, $e'=e$, and proceed to step 2.  Mark
  the columns in $A_j$ which are incident on non-zero positions in
  ${e}_1$.  That is, write 
  $$e_1=\sum_{i=1}^r a_i\otimes x_i,$$ where
  $a_i\in\mathbb{F}_2^{n_j}$, and $x_i\in \mathbb{F}_2^r$ with the
  only non-zero bit at position $i$.  Take $I_0$ the union of the
  supports of all vectors $a_i$.  Denote the corresponding submatrix
  of $A_j$ as $A^{(0)}_j=A_j[I_0]$; this is the generating matrix of a code
  $C_{A_j}$ punctured at the positions not in $I_0$.  Further, denote
  $A^{(0)}_{j+1}$ a transposed generating matrix of the code
  $\mathcal{C}_{A_{j+1}^T}$ shortened to $I_0$; it is obtained from a
  linear combination of columns of $A_{j+1}$ by dropping rows not in
  $I_0$.  
  
  By construction, $n^{(0)}_j\equiv |I_0|\le w_1$; since $w_1<d_j$,
  the homology group $H(A^{(0)}_j,A^{(0)}_{j+1})$ is trivial.  Now add
  a set of linearly independent columns from the remaining columns in
  $A_j$ into $A^{(0)}_j$ to get $A'_j=A_j[I_1]$, such that
  $|{I_1}|-|{I_0}|=\rank(A'_j)-\rank(A^{(0)}_j)$ and in addition
  $\rank(A'_j)=\rank(A_j)$. Similarly, denote $A'_{j+1}$ a transposed
  generating matrix of the code $\mathcal{C}_{A_{j+1}^T}$ shortened to
  $I_1$.
  %$A'_{j+1}$ can be obtained from a linear combination of $A_{j+1}$ by dropping rows not in $I_1$. 
  Then $H(A'_j,A'_{j+1})$ still has zero rank, and $H(A_{j-1},A_j)=H(A_{j-1},A'_j)$. 
  Use Eq.~(\ref{eq:tp-matrices}) to construct the
  corresponding matrices $C_j'$ and $C_{j+1}'$ and define the
  shortened vectors $e_1'=\sum_i a_i[I_1]\otimes x_i$,
  $e'=(e_1'|e_2)$.  Since we only removed zero positions, the new
  vector satisfies $C_j' (e')^T=0$.  Also, if there is a vector
  $\alpha'\in {\cal C}_{j+1}'$ such that $(e')^T=C_{j+1}'(\alpha')^T$,
  then necessarily $e^T=C_{j+1}\alpha^T$ with some
  $\alpha\in{\cal C}_{j+1}$.

    Since $H(A'_j,A'_{j+1})$ is trivial, in the
  next step we can construct a vector $\bar e'\simeq e'$ equivalent to
  $e'$ without worrying about the weight of its first block.

  Step 2: Consider the decomposition
  \begin{equation}
    {e}_2=\sum_{\ell=1}^c
    {f}_\ell\otimes{y}_\ell,\;f_\ell\in\mathbb{F}_2^{n_{j-1}},
    \label{eq:second-KP-representation}  
\end{equation}
where $ y_\ell\in \mathbb{F}_2^c$ has the only non-zero bit at
$\ell$. The identity $C_{j}'({e}')^T=0$ implies
$A_{j-1}{f}_\ell^T=0$ for any $1\le \ell\le c$.  For those $\ell $
where ${f}_{\ell}^T$ is linearly dependent with the columns of $A_j'$,
${f}_\ell^T= A_j'\alpha_\ell^T$ with some
$\alpha_\ell\in {\cal C}_j'=\mathbb{F}_2^{n_j'}$, render this vector
to zero by the equivalence transformation
$$
({e}')^T\to ({e}')^T+C_{j+1}' (0|\alpha_\ell\otimes{y}_\ell)^T.
$$ Such a transformation only
affects one vector ${f}_\ell$.  The resulting vector
$\bar{{e}}'=({e}_1'|{e}_2')$ has the second block of weight
$\wgt({e}_2')\le \wgt({e}_2)<d_{j-1}\delta$, it satisfies
$C_{j}'(\bar{{e}}')^T=0$, and in its block
representation (\ref{eq:second-KP-representation}) the remaining
non-zero vectors ${f}_\ell \in H(A_{j-1},A'_j)$ have weights $d_{j-1}$ or larger.

Step 3: For sure, there remains fewer than $\delta$ of non-zero
vectors ${f}_\ell$.  Thus, in a decomposition,
${e}_2'=\sum_{j=1}^{n_{j-1}}{z}_j\otimes {c}_j$, where
$z_j\in\mathbb{F}_2^{n_{j-1}}$ have the only non-zero bit at $j$, and
$c_j\in \mathbb{F}_2^c$, the union of supports of the vectors ${c}_j$,
$I_2$, has a length $c'\equiv |I_2|<\delta$.  Indeed, $I_2$ is
just the set of the indices $\ell$ corresponding to the remaining
non-zero vectors ${f}_\ell$.  Construct a matrix $P'=P[I_2]$ by
dropping the columns of $P$ outside of $I_2$.  Since there are fewer
than $\delta$ columns left, $c'<\delta$, the resulting classical code
contains no non-zero vectors, $c'=\rank P'$.  Construct the modified
matrices $C_j''$ and $C_{j+1}''$ and define the shortened vectors
$e_2''=\sum_{j=1}^{n_0}{z}_j\otimes {c}_j[I_2]$ and $e''=(e_1'|e_2'')$
such that $C_j''(e'')^T=0$.  Now, after we trimmed the columns of both
$A_j$ and of $P$, according to Eq.~(\ref{eq:stp-n-and-k}), the
homology group $H( C_j'', C_{j+1}'')$ is trivial.  This implies that
${e}''$ must be a linear combination of the columns of $ C_{j+1}''$,
that is, $(e'')^T=C_{j+1}'' \beta^T$, for some binary vector $\beta$.  

The transformation from $C_{j+1}'$ to $C_{j+1}''$ amounts to dropping
some columns in the right block of $C_{j+1}'$, and the matching rows
from the lower block.  The rows removed to obtain $e''$ correspond to
zero positions in $\bar e'$.  This implies that $\bar e'$ can be also
obtained as a linear combination of columns of $C_{j+1}'$,
$(\bar e')^T=C_{j+1}'(\beta')^T$.  Combined with the equivalence
transformation in Step 2, we get $(e')^T=C_{j+1}'(\alpha')^T$; the
construction of Step 1 then implies existence of
$\alpha\in{\cal C}_{j+1}$ such that $e^T=C_{j+1}\alpha^T$ for the
original two-block vector $e=(e_1|e_2)$.  Thus, any such $e$ with
block weights $w_1<d_{j}$ and $w_2<d_{j-1}\delta$ which satisfies
$C_j{e}^T=0$ is necessarily a linear combination of the columns of
$C_{j+1}$.  This guarantees $d_j'\ge \min(d_{j},d_{j-1}\delta)$.

To complete the proof, consider the case (ii).  Here, step 1 can be
omitted; the matrices resulting from steps 2 and 3 alone would give
trivial homology group, regardless of the weight $\wgt({e}_1)$
of the first block.  Thus, in this case we get the lower bound
$d_j'\ge d_{j-1}\delta$.
\end{proof}
Let us now consider tensor products of several $1$-complexes.  Basic
parameters such as space dimensions, row and column weights, or
homology group distances do not depend on the order of the terms in
the product.  Further, if the matrices used to construct one-complexes
are $(\upsilon,\omega)$-sparse, that is, their column and row weights
do not exceed $\upsilon$ and $\omega$, respectively, the matrices in
the resulting $m$-chain complex are $(m\upsilon,m\omega)$-sparse.

As the first example, consider an $r\times c$ full-row rank binary
matrix $P$ with $r<c$, and assume that a binary code
$\mathcal{C}_{P}^\perp$ with the parity check $P$ has distance
$\delta$.  The $1$-complex $\mathcal{K}\equiv \mathcal{K}(P)$ has two
non-trivial spaces of dimensions $r$ and $c$; the corresponding
homology groups have ranks $0$, $\kappa$ and the distances $\infty$,
$\delta$.  The $1$-complex $\tilde{K}\equiv\mathcal{K}({P}^T)$
generated by the transposed matrix has equivalent spaces taken in the
opposite order, with the same homology group ranks, but the distances
are now $1$ and $\infty$, respectively.  It is easy to see that in any
chain complex constructed as tensor products of $\mathcal{K}$ and/or
$\tilde{\mathcal{K}}$, there is going to be only one homology group
with a non-zero rank.  Since order of the products is not important,
we will write these as powers.  For $(a+b)$-complex 
$\mathcal{K}^{(a,b)}\equiv\mathcal{K}^{\times a}\times
\tilde{\mathcal{K}}^{\times b}$, the only non-trivial homology group
is $H_a(\mathcal{K}^{(a,b)})$; the corresponding space has the
dimension
$$
n_a(\mathcal{K}^{(a,b)})=\sum_{i=0}^a c^{2i} r^{a+b-2i}{a\choose
  i}{b\choose i}<(r+c)^{a+b},
$$
homology group rank $\kappa^{a+b}$, and distance $\delta^a$.  The
corresponding quantum CSS code has the conjugate distances $\delta^a$
and $\delta^b$, and its stabilizer generators have weights not
exceeding $(a+b)\max(\omega,\upsilon)$.  Good weight-limited classical
codes with finite rates $\kappa/c$ and finite relative distances
$\delta/c$ can be obtained from ensembles of large random
matrices\cite{Gallager-1962,Gallager-book-1963,Litsyn-Shevelev-2002,
  Richardson-Shokrollahi-Amin-Urbanke-2001}.  Any of these can be used
in the present construction.  Then, for any pair $(a,b)$ of natural
numbers, we can generate weight-limited q-LDPC codes with finite rates
and the distances $d_X=\delta^a$, $d_Z=\delta^b$ whose
product scales linearly with the code length.  QHP codes are a special
case of this construction with $a=b=1$. 

Unlike in the case of QHP codes, with any $a>1$, $b>1$, the rows of
matrices $G_X=K_{a}\equiv K_a(\mathcal{K}^{(a,b)})$, $G_Z=K_{a+1}^T$
satisfy a large number of linear relations resulting from the
orthogonality with the matrices $K_{a-1}$ and $K_{a+2}$, respectively.
These can be used to correct syndrome measurement errors.  Even though
the resulting syndrome codes do not have large distances (with a
finite probability some errors remain), the use of such codes in
repeated measurement setting could simplify the decoding and/or
improve the decoding success probability in the case of adversarial
noise\cite{Campbell-2018}.  Such improvements with stochastic noise
have been demonstrated numerically in the case of $4D$ toric codes in
Ref.~\onlinecite{Breuckmann-Duivenvoorden-Michels-Terhal-2017}.

In conclusion, we derived an explicit expression for the distances of
the homology groups in a tensor product of two chain complexes, in the
special case where one of the complexes has length two.  Immediate use
of this result is in theory of quantum LDPC codes.  Our result greatly
extends the family of QHP codes whose parameters are known explicitly.
Higher-dimensional QHP codes can be especially useful in
fault-tolerant quantum computation, to optimize repeated syndrome
measurement in the presence of measurement errors.
 
In addition, we believe that the lower bound on the distance in
Theorem \ref{th:lower-distance-bnd} can be extended to a general
product of two chain complexes.  Indeed, Eq.~(\ref{eq:tp-boundary})
implies that the corresponding block matrices have at most two
non-zero blocks in each row and each column; similar steps can be used
in a proof.  If this is the case, the r.h.s.\ in
Eq.~(\ref{eq:tp-dj-upper-bnd}) would give explicitly the distances,
not just an upper bound.  Such a result could have substantial
applications in many areas of science where homology is used.

\begin{acknowledgments}
  LPP is grateful to Jean-Pierre Tilich for illuminating discussions,
  and to the Institute Henri Poincar\'e for hospitality.  This
  research was supported in part by the NSF Division of Physics via
  Grants No.\ 1416578 and 1820939.
\end{acknowledgments}

\bibliography{lpp,qc_all,more_qc,ldpc,linalg,percol,sg}
\end{document}